\documentclass {ws-dmaa}
\usepackage{epstopdf}
\usepackage{caption}

\makeatletter
\def\imod#1{\allowbreak\mkern10mu({\operator@font mod}\,\,#1)}
\makeatother

\begin{document}

\markboth{S. Chen, Y. Chen, and Q. Yang}
{Randomized Testing of $q$-Monomials}

\title{Towards Randomized Testing of $q$-Monomials in Multivariate Polynomials}

\author{Shenshi Chen and Yaqing Chen}
\address{Department of Computer Science,\\
University of Texas-Pan American,\\
Edinburg, TX 78539, USA\\
\email{schen@broncs.utpa.edu}}


\author{Quanhai Yang}
\address{College of Information Engineering,\\
Northwest A\&F University,\\
Yangling, Shaanxi 712100, China\\}

\maketitle

\begin{abstract}
Given any fixed integer $q\ge 2$, a $q$-monomial is of the format
$\displaystyle x^{s_1}_{i_1}x^{s_2}_{i_2}\cdots x_{i_t}^{s_t}$ such that
$1\le s_j \le q-1$, $1\le j \le t$. $q$-monomials are natural generalizations of
multilinear monomials. Recent research on testing multilinear
monomials and $q$-monomials for prime $q$ in multivariate polynomials relies on
the property that $Z_q$ is a field when $q\ge 2 $ is prime.
When $q>2$ is not prime, it remains open whether the problem of testing $q$-monomials
can be solved in some compatible complexity. In this paper, we present
a randomized $O^*(7.15^k)$ algorithm for testing $q$-monomials of degree $k$
that are found in a multivariate polynomial that is represented by a tree-like circuit
with a polynomial size,
thus giving a positive, affirming answer to the above question.
Our algorithm works regardless of the primality of $q$ and
improves upon the time complexity of
the previously known algorithm for testing $q$-monomials for prime $q>7$.
\end{abstract}

\section{Introduction}

\subsection{Background}

Recently, significant efforts have been made towards studying the problem of
testing monomials in multivariate polynomials \cite{koutis08,williams09,Bjorklund2010,chen12a,chen12b,chen11,chen11a,chen11b},
with the central question consisting of whether a multivariate polynomial represented by
a circuit (or even simpler structure) has a multilinear (or some specific) monomial in
its sum-product expansion. This question can be answered straightforwardly when
the input polynomial has been expanded into a sum-product representation,
but the dilemma, though, is that obtaining such a representation generally requires exponential time.
The motivation and necessity of studying the monomial testing problem can be
clearly understood from its connections to various critical problems in computational complexity
as well as the possibilities of applying algebraic properties of polynomials
to move forward the research on those critical problems (see, e.g., \cite{chen11}).

Historically, polynomials and the studies thereof have, time and again,
contributed to many advancements in
theoretical computer science research.
Most notably, many major breakthroughs in complexity theory would not have
been possible without the invaluable roles played by low degree polynomial testing/representing
and polynomial identity testing.
For example, low degree polynomial testing was involved in the proof of the PCP Theorem,
the cornerstone of the theory of computational hardness of
approximation and the culmination of a long line of research on IP
and PCP (see, Arora {\em et al.} \cite{arora98} and Feige {\em et
al.} \cite{feige96}). Polynomial identity testing has been
extensively studied due to its role in various aspects of
theoretical computer science (see, for example, Kabanets and Impagliazzo \cite{kabanets03})
and its applications in various fundamental results such as Shamir's
IP=PSPACE \cite{shamir92} and the AKS Primality Testing
\cite{aks04}. Low degree polynomial representing
\cite{minsky-papert68} has been sought after in order to prove
important results in circuit complexity, complexity class
separation and subexponential time learning of Boolean functions
(see, for examples, Beigel \cite{beigel93}, Fu\cite{fu92}, and
Klivans and Servedio \cite{klivans01}).
Other breakthroughs in the field of algorithmic design have also been achieved
by combinations of randomization and algebrization.
Randomized algebraic techniques have led to the randomized algorithms of time $O^*(2^k)$ for the {\sc $k$-path}
problem and other problems \cite{koutis08,williams09}. Another recent
seminal example is the improved randomized $O(1.657^n)$ time algorithm for the Hamiltonian
path problem by Bj{\"o}rklund~\cite{Bjorklund2010}. This algorithm provided a positive answer to
the question of whether the Hamiltonian path problem can be solved in time $O(c^n)$ for some constant
$0<c < 2$, a challenging problem that had been open for
half of a century. Bj{\"o}rklund {\em et al.} further extended the above randomized algorithm
to the $k$-path testing problem with $O^*(1.657^k)$ time complexity \cite{Bjorklund2010b}.
Very recently, those two algorithms  were simplified by
Abasi and Bshouty \cite{bshouty13}. These are just a few
examples and a survey of related literature is beyond
the scope of this paper.

\subsection{The Related Work}

The problem of testing multilinear
monomials in multivariate polynomials
was initially exploited by Koutis \cite{koutis08}
and then by Williams \cite{williams09} to design randomized
parameterized algorithms for the $k$-path problem.
Koutis \cite{koutis08} initially developed an innovative group algebra approach to
testing multilinear monomials with odd coefficients in the sum-product expansion of
any given multivariate polynomial. Williams \cite{williams09} then further connected
the polynomial identity testing problem to multilinear monomial testing and devised an algorithm
that can test multilinear monomials with odd or even coefficients.

The work by Chen {\em et al.} \cite{chen12a,chen12b,chen11,chen11a,chen11b}
aimed at developing a theory of testing monomials in
multivariate polynomials in the context of a computational complexity study.
The goal was to investigate the various complexity aspects
of the monomial testing problem and its variants.

Initially, Chen and Fu
\cite{chen11} proved a series of foundational results,
beginning with the proof that the multilinear
monomial testing problem for $\Pi\Sigma\Pi$ polynomials is
NP-hard, even when each factor of the given polynomial has at most three product terms and each
product term has a degree of at most $2$. These results have built a base upon which
further study of testing monomials can continue.

Subsequently, Chen {\em et al.} \cite{chen11b} (see, also,\cite{chen12b})
studied the generalized $q$-monomial testing problem. They proved that when $q\ge 2$ is prime,
there is a randomized $O^*(q^k)$ time algorithm
for testing $q$-monomials of degree $k$ with coefficients $\not=0\imod q$ in
an arithmetic circuit representation of a multivariate polynomial
which can then be derandomized into
a deterministic $O^*((6.4p)^k)$ time
algorithm when the underlying graph of the circuit is a tree.

In the third paper, Chen and Fu \cite{chen12a} (and \cite{chen11a})
turned to finding the coefficients
of monomials in multivariate polynomials.
Naturally, testing for the existence of any given monomial in a
polynomial can be carried out by computing the coefficient of that
monomial in the sum-product expansion of the polynomial. A zero
coefficient means that the monomial is not present in the polynomial,
whereas a nonzero coefficient implies that it is present. Moreover,
they showed that coefficients of monomials in a polynomial have their own
implications and are closely related to core problems in
computational complexity.

\subsection{Contribution and Organization}

Recent research on testing multilinear
monomials and $q$-monomials for prime $q$ in multivariate polynomials relies on
the property that $Z_2$ and  $Z_q$ are fields only when $q> 2 $ is prime.
When $q>2$ is not prime, $Z_q$ is no longer a field, hence
the group algebra based approaches in \cite{koutis08,williams09,chen11b,chen12b}
are not applicable to cases of non-prime $q$.
It remains open whether the problem of testing $q$-monomials
can be solved in some compatible complexity for non-prime $q$.
Our contribution in this paper is
a randomized $O^*(7.15^k s^2(n))$ algorithm for testing $q$-monomials of degree $k$
in a multivariate polynomial represented by a tree-like circuit of size $s(n)$,
thus giving an affirming answer to the above question.
Our algorithm works for both prime $q$ and non-prime $q$ as well.
Additionally, for prime $q>7$,
our algorithm provides us with some substantial improvement on the time complexity of
the previously known algorithm \cite{chen11b,chen12b} for testing $q$-monomials. 

The rest of the paper is organized as follows. In Section 2, we
introduce the necessary notations and definitions. In Section 3,
we examine three examples to understand the
difficulty to transform $q$-monomial testing to multilinear monomial testing.
In Section 4, we propose a new method for reconstructing a given circuit and
a technique to replace each occurrence of a variable with a randomized linear sum of
$q-1$ new variables.
We show that, with the desired probability,
the reconstruction and randomized replacements help transform
the testing of $q$-monomials in any polynomial represented by a tree-like circuit to
the testing of multilinear monomial in a new polynomial.
We design a randomized $q$-monomial testing algorithm in Section 5 and
conclude the paper in Section 6.

\section{Notations and Definitions}

 For $1\le i_1 < \cdots <i_k \le n$, $\pi =x_{i_1}^{s_1}\cdots
x_{i_t}^{s_t}$ is called a monomial. The degree of $\pi$, denoted
by $\mbox{deg}(\pi)$, is $\sum\limits^t_{j=1}s_j$. $\pi$ is multilinear,
if $s_1 = \cdots = s_t = 1$, i.e., $\pi$ is linear in all its
variables $x_{i_1}, \dots, x_{i_t}$. For any given integer $q\ge
2$, $\pi$ is called a $q$-monomial if $1\le s_1, \dots, s_t \le q-1$.
In particular, a multilinear monomial is the same as  a $2$-monomial.

An arithmetic circuit, or circuit for short, is a directed acyclic
graph consisting of $+$ gates with unbounded fan-ins, $\times$ gates with two
fan-ins, and terminal nodes that correspond to variables. The size,
denoted by $s(n)$, of a circuit with $n$ variables is the number
of gates in that circuit. A circuit is considered a tree-like circuit 
 if the fan-out of every gate is at most one, i.e.,
the underlying directed acyclic graph that excludes all the terminal nodes is a tree.
In other words, in a tree-like circuit, only the terminal nodes can have more than one fan-out
(or out-going edge).

Throughout this paper, the $O^*(\cdot)$ notation is used to
suppress $\mbox{poly}(n,k)$ factors in time complexity bounds.

By definition, any polynomial $F(x_1,\dots,x_n)$ can be expressed
as a sum of a list of monomials, called the sum-product expansion.
The degree of the polynomial is the largest degree of its
monomials in the expansion. With this expanded expression, it is trivial to
see whether $F(x_1,\dots,x_n)$ has a multilinear monomial, or a
monomial with any given pattern. Unfortunately, such an expanded expression is
essentially problematic and infeasible due to the fact that a
polynomial may often have exponentially many monomials in its
sum-product expansion.

In general, a polynomial $F(x_1,\dots,x_n)$ can be represented by
a circuit. 
 This type of representation is simple and compact and
may have a substantially smaller size polynomially in $n$,
when compared to the number of all monomials in its sum-product
expansion. Thus, the challenge then is to test whether $F(x_1,\dots,x_n)$
has a multilinear (or some other desired) monomial efficiently,
without expanding it into its sum-product representation.

For any given $n\times n$ matrix ${\cal A}$,
let $\mbox{perm}({\cal A})$ denote the
permanent of ${\cal A}$ and $\mbox{det}({\cal A})$ the
determinant of ${\cal A}$.

For any integer $k \ge 1$, we consider the group
$Z^k_2$ with the multiplication $\cdot$ defined as follows. For
$k$-dimensional column vectors $\vec{x}, \vec{y} \in Z^k_2$ with
$\vec{x} = (x_1, \ldots, x_k)^T$ and $\vec{y} = (y_1, \ldots,
y_k)^T$, $\vec{x} \cdot \vec{y} = (x_1+y_1, \ldots, x_k+y_k)^T.$
$\vec{v}_0=(0, \ldots, 0)^T$ is the zero element in the group.
For any field ${\cal F}$, the group algebra ${\cal F}[Z^k_2]$ is defined as
follows. Every element $u \in {\cal F}[Z^k_2]$ is a linear addition of the form
\begin{eqnarray}\label{exp-2}
 u &=& \sum_{\vec{x}_i\in Z^k_2,~ a_{i}\in {\cal F}} a_{i} \vec{x}_i.
\end{eqnarray}
For any element
$v = \sum\limits_{\vec{x}_i\in Z^k_2,~ b_{i}\in {\cal F}} b_{i}
\vec{x}_i$,
we define
\begin{eqnarray}
 u + v  &=& \sum_{a_{i},~ b_{i}\in {\cal F},~  \vec{x}_i\in Z^k_2}  (a_i+b_i)
 \vec{x}_i, \  \mbox{and} \nonumber\\
u \cdot v &=& \sum_{a_i,~ b_j\in {\cal F},~ \mbox{ and }~\vec{x}_i,~ \vec{y}_j\in Z^k_2}  (a_i b_j)
(\vec{x}_i\cdot \vec{y}_j). \nonumber
\end{eqnarray}
For any scalar $c \in {\cal F}$,
\begin{eqnarray}
 c u &=& c \left(\sum_{\vec{x}_i\in Z^k_p, \ a_i\in {\cal F}} a_{i} \vec{x}_i\right)
 = \sum_{\vec{x}_i\in Z^k_2,\  a_{i}\in {\cal F}} (c a_{i})\vec{x}_i. \nonumber
\end{eqnarray}
The zero element in the group algebra $\displaystyle {\cal F}[Z^k_2]$ is
$\displaystyle {\bf  0} = \sum_{\vec{v}} 0\vec{v}$, where $0$ is the zero element in ${\cal F}$
and $\vec{v}$ is any vector in $\displaystyle Z_2^k$. For example,
${\bf  0} = 0\vec{v_0} = 0\vec{v}_1 + 0\vec{v}_2 + 0\vec{v}_3$,
for any $\displaystyle \vec{v}_i \in Z^k_2$, $1\le i\le 3$.
The identity element in the group algebra $\displaystyle {\cal F}[Z^k_2]$ is
$ {\bf 1} = 1 \vec{v}_0 =  \vec{v}_0$, where $1$ is the identity element in ${\cal F}$.
For any vector $\vec{v} =(v_1, \ldots, v_k)^T \in Z_2^k$, for
$i\ge 0$, let
$\displaystyle (\vec{v})^i = (i v_1, \ldots, i v_k)^T.$
When the field ${\cal F}$ is  $Z_2$ with respect to $\imod 2$ operation, for any
$x,y\in Z_2$, $x y $ and
$x+y$ stands for  $x y \imod 2$ and $x+y \imod 2$,
respectively.
 In particular, in the group algebra $Z_2[Z_2^k]$,
for any $\vec{z}\in Z_2^k$ we have $(\vec{v})^0 = (\vec{v})^2 =
\vec{v}_0.$

\section{$q$-Monomials, Multilinear Monomials and Plus Gates} 

As we pointed out before, group algebra based algorithms \cite{koutis08,williams09,chen11b,chen12b} cannot be called upon to
test $q$-monomials when $q$ is not prime, because $Z_q$ is not a field. Hence, in such a case the algebraic foundation
for applying those algorithms is no longer available.

It seems quite hopeful that there might be a way to transform the problem of testing $q$-monomials into the problem of testing multilinear monomials and thus utilize the existing techniques for the latter problem to solve the former problem. One plausible strategy to accomplish such a transformation is to replace each
variable $x$ in a given multivariate polynomial by a sum $y_1+y_2+\cdots+y_{q-1}$ of $q-1$ new variables.
Ideally, such replacements should result in a multilinear monomial in the new polynomial that corresponds to the given $q$-monomial in the original
polynomial and vice versa, thereby allowing the multilinear monomial testing algorithm based on
some group algebra over a field of characteristic $2$~\cite{koutis08,williams09} to be adopted for the testing of multilinear monomials in the new polynomial. Unfortunately, some careful analysis will reveal that this approach has, as exhibited in Example \ref{ex-1}, a profound technical barrier that prevents us from applying those mulilinear monomial testing algorithms.

\begin{example}\label{ex-1}
Consider a simple $4$-monomial $\pi = x^3$ of degree $3$.
Replacing $x$ with $y_1+y_2+y_3$ in $\pi$ results in
\begin{eqnarray}
r(\pi) &=& (y_1+y_2+y_3)^3 \nonumber \\
&= & y_1^3+y_2^3+y_3^3 + 3y_1^2y_2 + 3y_1y^2_2 +3y_2^2y_3 + 3y_2y^2_3
 + 3y_1^2y_3 + 3y_1y^2_3 \nonumber \\
 & & +~ 6 y_1y_2y_3. \nonumber
\end{eqnarray}
\end{example}

$r(\pi)$ has one and only one degree $3$ multilinear monomial $\pi'=y_1y_2y_3$.
It is unfortunate that the coefficient $c(\pi')$ of $\pi'$ is $6$, an even number.
When applying the group algebra
based multilinear monomial testing algorithms to $r(\pi)$ over the field $Z_2$
with respect to $(\bmod~ 2)$ operation, the even coefficient $c(\pi')$ will help eliminate
$\pi'$ from $r(\pi)$. Hence, we are unable to find the existence of any multilinear monomials
in the sum-product expansion of $r(\pi)$.

Knowing that the above example can be generalized to arbitrary $q$-monomials for $q>2$,
we have to design an innovative replacement technique so that certain multilinear monomials
in the new polynomial will survive the elimination by the $(\bmod~ 2)$ operation over $Z_2$, or
by the characteristic 2 property over any field of characteristic 2.
Specifically, we have to ensure, with complete or desired probabilistic certainty,
that a given $q$-monomial $\pi$ with coefficient $c(\pi)$ in the original polynomial will correspond to
one or a list of {\em "distinguishable"} multilinear monomials with odd coefficients
in the derived polynomial, regardless of the parity of $c(\pi)$.

\begin{figure}[h]
\centering%
\begin{minipage}[b]{.4\textwidth}
    \centering%
    \includegraphics[width=.3\linewidth]{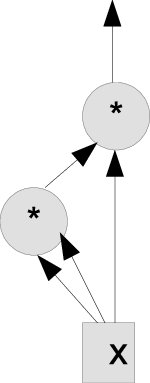}
    \captionof{figure}{A Circuit for $\pi = x^3$}
    \label{fig1}
\end{minipage}%
\begin{minipage}[b]{.6\textwidth}
    \centering%
    \includegraphics[width=.6\linewidth]{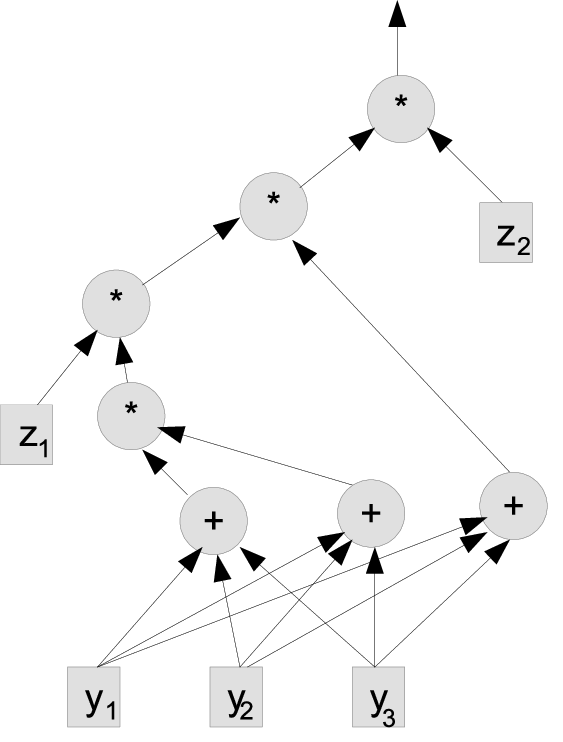}
    \captionof{figure}{The Expanded Circuit for $\pi = x^3$}
    \label{fig2}
\end{minipage}%
\end{figure}



When group algebraic elements are selected to replace variables
in the input polynomial, the polynomial might become zero due to mutual annihilation
of the results from a list of multilinear monomials with odd coefficients.
Koutis \cite{koutis08} proved that when those group algebraic elements are uniform random,
with a probability at least $\frac{1}{4}$,
the input polynomial that has multilinear monomials with odd coefficients
will not become zero, even if mutual annihilation
of the results from a list of multilinear monomials with odd coefficients may happen.

Williams \cite{williams09} introduced a new variable for each $\times$ gate in the representative circuit
for the input polynomial that can help avoid the aforementioned mutual annihilation. In essence,
the new variables added for the $\times$ gates can help generate one or a list of "distinguishable"
multilinear monomials with odd coefficients in the derived polynomial, no matter whether the coefficient of
the original multilinear monomial is even or odd.  However, this approach cannot help
resolve the $q$-monomial testing problem, due to possible implications of $+$ gates.

In order to understand the above situation,
let us examine Example 1 again. Following Williams's algorithm, we first reconstruct the circuit in Figure~\ref{fig1}.
The expanded circuit, after the replacement of $x$ by $y_1+y_2+y_3$
along with the addition of new variables $z_1$ and $z_2$
for the two respective $+$ gates, is shown in Figure~\ref{fig2}.
The coefficient for the only multilinear monomial $y_1y_2y_3$ produced by the new circuit is
$6z_1z_2$, which is even and thus helps annihilate $y_1y_2y_3$ with respect to $(\bmod~ 2)$ operation
or in general the characteristic 2 property of the underlying field.

The following two examples provide us with more evidences that there are technical difficulties in
dealing with possible implications of $+$ gates.

\begin{figure}[h]
\centering
\begin{minipage}[b]{.4\textwidth}
    \centering
    \includegraphics[width=.7\linewidth]{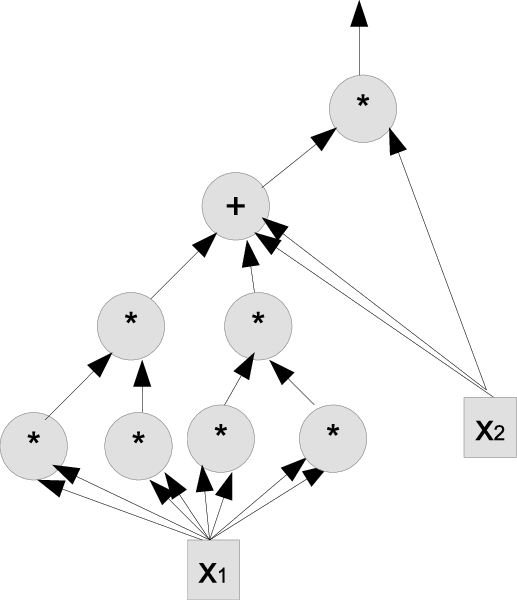}
    \captionof{figure}{A Circuit for $F(x_1,x_2)$}
    \label{fig3}
\end{minipage}%
\begin{minipage}[b]{.6\textwidth}
    \centering
    \includegraphics[width=.8\linewidth]{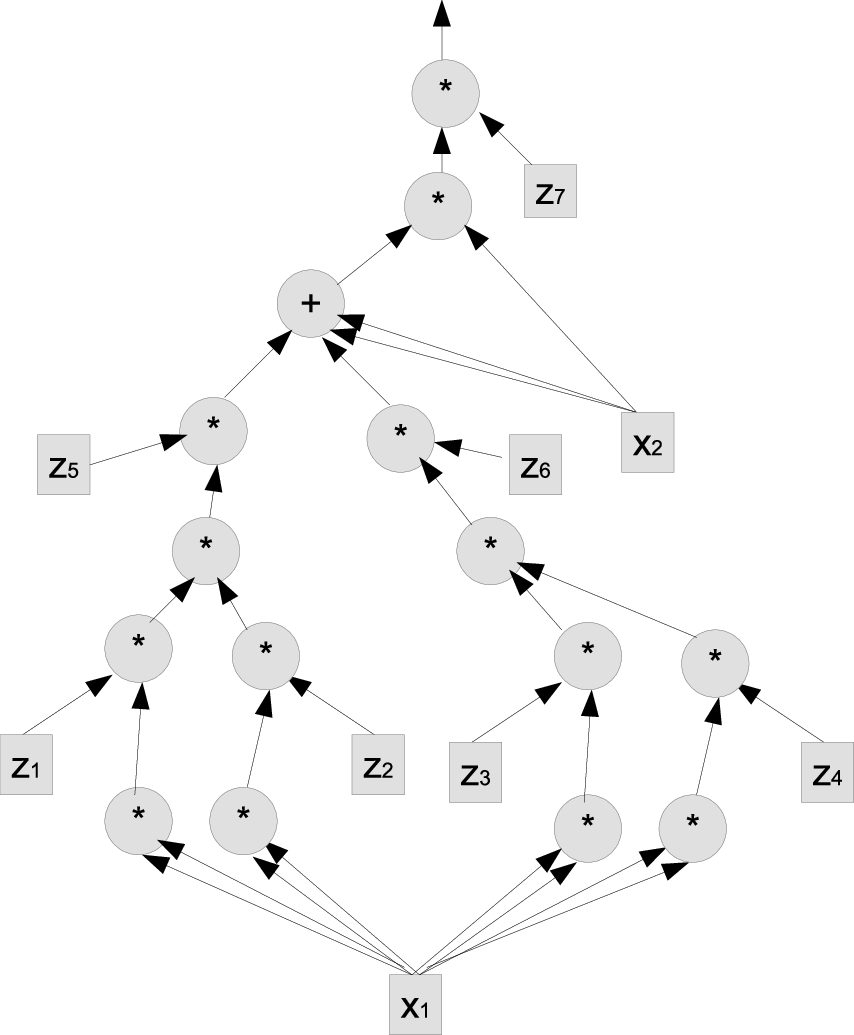}
    \captionof{figure}{The Reconstructed Circuit for $F(x_1,x_2)$}
    \label{fig4}
\end{minipage}
\end{figure}


\begin{example}\label{ex-2}
Let $F(x_1,x_2) = 2x_1^4x_2 + 2x_2^2$ as represented by the circuit in Figure~\ref{fig3}.
$F$ has one $5$-monomial $\pi_1 = x_1^4x_2$ and one $3$-monomial $\pi_2 = x_2^2$, each of which has a coefficient $2$.
\end{example}

When one follows the approach by Williams \cite{williams09} to add, for each $\times$ gate in Figure~\ref{fig3},
a new $\times$ gate that multiplies the output of this gate with a new variable, then one obtains a new circuit in Figure~\ref{fig4}
that computes
$$F'(z_1,z_2,\ldots,z_7,x_1,x_2) = z_1z_3z_5z_7x_1^4x_2 + z_3z_4z_6z_7x^4_1x_2 + 2z_7x_2^2.$$
Although $2x_1^4x_2$ in $F$
is spilt into two distinguishable occurrences that have respective unique coefficients  $z_1z_3z_5z_7$ and  $z_3z_4z_6z_7$,
yet $2x_2^2$ in $F$ corresponds to $2z_7x_2^2$  that has an even coefficient $2z_7$.

In particular, the implications of $+$ gates on testing multilinear monomials can be seen from the following example.

\begin{example}\label{ex-3}
Let $G(x_1,x_2,x_3) = 2x_1^2x_3 + 2x_2x_3$. Changing the terminal node $x_2$ to $x_3$ for the top $\times$ gate
in Figure~\ref{fig3} (respectively, for the top second $\times$ gate in Figure~\ref{fig4} gives a circuit  to compute $G$  (respectively, $G'$).
\end{example}

Like in Example 2, $G'(z_1,z_2,z_3,x_1,x_2, x_3) = z_1z_3z_5z_7x_1^4x_3 + z_3z_4z_6z_7x^4_1x_3 + 2z_7x_2x_3$.
Here, $2x_1^4x_3$ is spilt into two distinguishable occurrences that have unique coefficients   $z_1z_3z_5z_7$ and  $z_3z_4z_6z_7$, respectively.
However, the only multilinear monomial $2x_2x_3$ in $G$ corresponds to  $2z_7x_2x_3$ that has an even coefficient
$2z_7$. Therefore, this multilinear monomial cannot be detected by Williams' algorithm.

Example 3 exhibits that there is a flaw
in the circuit reconstruction by Williams \cite{williams09}: Introducing a new variable to multiply the output of every
$\times$ gate is not sufficient to overcome the difficulty that may possibly be caused by $+$ gates.

\section{Circuit Reconstruction and A Transformation}

In this section, we shall design a new method to reconstruct a given circuit and a randomized variable replacement technique
so that we can transform, with some desired success probability, the testing of $q$-monomials  to the testing of  multilinear monomials.

To simplify presentation, we assume from now on through the rest of the paper that if any given
polynomial has $q$-monomials in its sum-product expansion, then the degrees
of those multilinear monomials are at least
$k$ and one of them has exactly a degree of $k$. This assumption is feasible,
because when a polynomial has $q$-monomials of degree $< k$, e.g., the least degree of those is $\ell$ with $1\le \ell < k$,
then we can multiply
the polynomial by a list of $k-\ell$ new variables so that the resulting polynomial
will have $q$-monomials with degrees satisfying the aforementioned assumption.

\subsection{Circuit Reconstruction}\label{CR}

For any given polynomial $F(x_1,x_2,\ldots,x_n)$ represented by a tree-like circuit ${\cal C}$ of size $s(n)$,
we first reconstruct the circuit ${\cal C}$ in three steps as follows.

{\bf Eliminating redundant $+$ gates.} Starting with the root gate, check to see whether a
$+$ gate receives input from another $+$ gate. If a $+$ gate $g$ receives input from a $+$ gate $f$, which
receives inputs from gates $f_1,f_2,\ldots,f_{s}$ and/or terminal nodes $u_1,u_2,\ldots,u_t$, then delete $f$ and let the gate
$g$ to receive inputs directly from  $f_1,f_2,\ldots,f_{s}$ and/or $u_1,u_2,\ldots,u_t$. Repeat this process until
there are no more $+$ gates receiving input from another $+$ gate.

Note that we consider tree-like circuits only. Since each gate of such a circuit has at most one output,
the above eliminating process will not increase the size of  the circuit.

{\bf Duplicating terminal nodes.} For each variable $x_i$,
if $x_i$ is the input to a list of gates $g_1, g_2, \ldots, g_{\ell}$, then create $\ell$
terminal nodes $u_1, u_2, \ldots, u_{\ell}$ such that each of them represents
a copy of the variable $x_i$ and $g_j$ receives input from $u_j$, $1\le j\le \ell$.

Let ${\cal C}^*$ denote the reconstructed circuit after the above two reconstruction steps.
Since the original circuit ${\cal C}$ is tree-like, the underlying graph of ${\cal C}^*$, including all the terminal nodes,
is  a tree. Such a tree structure implies the following simple facts:
\begin{itemize}
\item There is no duplicated occurrence of any input variable along any path from
 the root to a terminal node.
 \item Every occurrence of each variable $x_i$
 in the sum-product expansion of $F$ is represented by a terminal node for $x_i$.
 \item The size of the new circuit is at most $n s(n)$.
 \item Any $+$ gate will receive input from $\times$ gates and/or terminal nodes.
\end{itemize}

\begin{figure}[h]
\centering
\includegraphics[width=.55\linewidth]{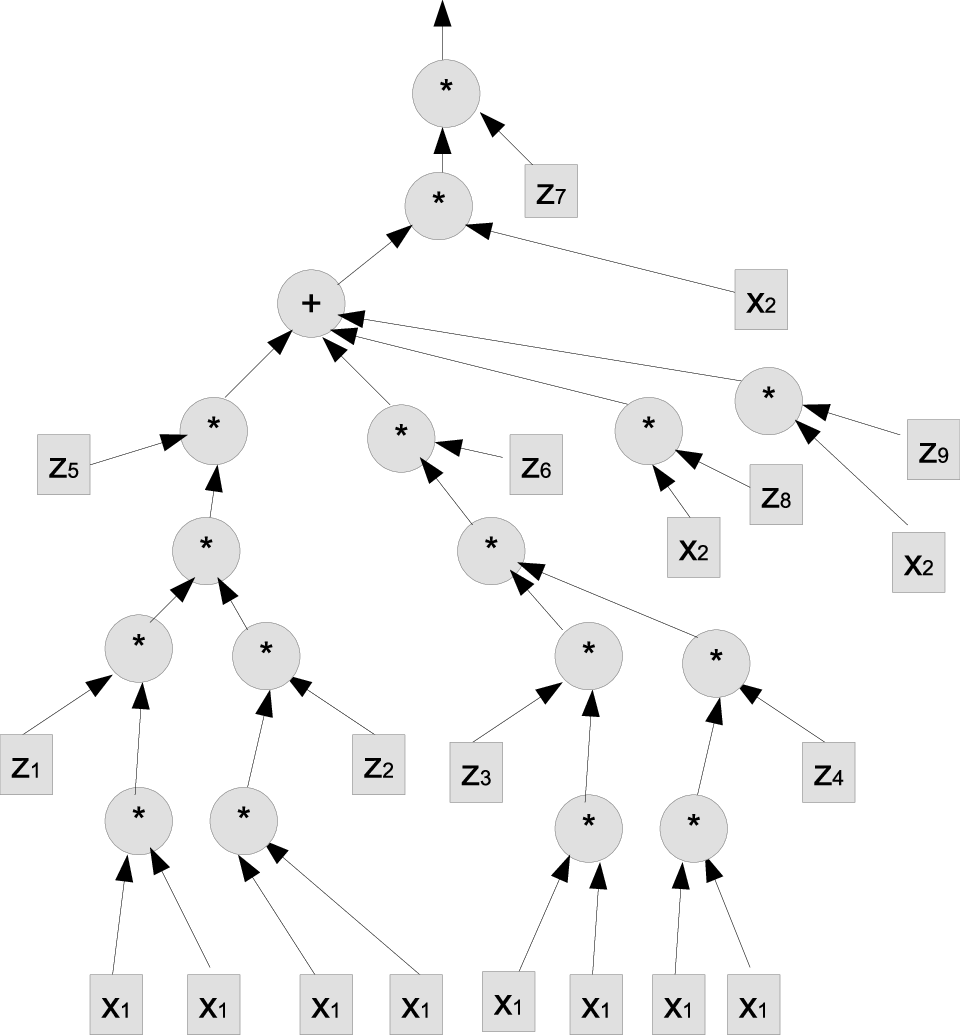}
\caption{The New Circuit for $F(x_1,x_2) = 2x_1^4x_2 + 2x_2^2$}
\label{fig5}
\end{figure}


{\bf Adding new variables for $\times$ gates and for those terminal nodes that directly connect to $+$ gates.}
Having completed the reconstruction for ${\cal C}^*$,
we then expand it to a new circuit ${\cal C'}$ as follows. For each
$\times$ gate $g_i$ in ${\cal C}^*$, we attach a new $\times$ gate $g'_i$ that
multiplies the output of $g_i$ with a new variable $z_i$, and feed
the output of $g'_i$ to the gate that reads the output of $g_i$.
Here, the way of introducing new variables for $\times$ gates follows
what is done by Williams in \cite{williams09}.
However, in addition to these new $z$-variables, we may need to introduce additional variables for $+$ gates.
Specifically, for each $+$ gate $f$ that receives inputs from terminal nodes
$u_1,u_2,\ldots,u_t$, we add a $\times$ gate $f_j$ and have it to receive inputs from $u_j$ and a new variable $z_j$
and then feed its output to $f$, $1\le j\le t$. Note that $f$ may receive input from $\times$ gates
but no new gates are needed for those gates with respect to $f$.

Assume that a list of $h$ new $z$-variables
$z_1, z_2, \ldots, z_h$ have been introduced into the circuit ${\cal C'}$.
Let $F'(z_1, z_2, \ldots, z_h, x_1, x_2,
\ldots, x_n)$ be the new polynomial represented by ${\cal C'}$.

In Figure~\ref{fig5}, we show the reconstructed circuit for the one in Figure~\ref{fig3} that represents
$F(x_1,x_2) = 2x_1^4x_2 + 2x_2^2$. By this new circuit,
$$F'(z_1,z_2,\ldots,z_9,x_1,x_2) = z_1z_2z_5z_7x_1^4x_2 + z_3z_4z_6z_7x^4_1x_2 + z_7z_8x_2^2 + z_7z_9x_2^2.$$
As expected, not only is $2x_1^4x_2$ in $F$
split into two distinguishable occurrences that have unique coefficients  $z_1z_2z_5z_7$ and  $z_3z_4z_6z_7$,
but also $2x_2^2$ in $F$ is split into two distinguishable occurrences that have unique coefficients  $z_7z_8$ and  $z_7z_9$.
Notably, those four coefficients are multilinear monomials of $z$-variables and each has an odd scalar coefficient 1.

\begin{lemma}\label{rtm-lem1}
$F(x_1,x_2,\ldots,x_n)$ has a monomial $\pi$ of degree $k$  in its sum-product expansion if and only if there is
a monomial $\alpha \pi$ in the sum-product expansion of $F'(z_1, z_2, \ldots, z_h, x_1, x_2,\ldots, x_n)$
such that $\alpha$ is a multilinear monomial of $z$-variables with degree $\le 2k-1$.
Furthermore, if $F'$ has two products $\alpha_1\pi$ and $\alpha_2\pi$ in its sum-product expansion, then
we have $\alpha_1 \not= \alpha_2$, where $\alpha_1$ and $\alpha_2$ are products of $z$-variables;
and any two different monomials of $x$-variables
in $F'$ will have different coefficients that are products of $z$-variables.
\end{lemma}

\begin{proof}

By the reconstruction processes, ${\cal C^*}$ computes exactly the same polynomial $F$.
If $F$ has a monomial $\pi$ of degree $k$,
then let ${\cal T}$ be the subtree of ${\cal C^*}$ that generates the monomial $\pi$,
and ${\cal T'}$ be the corresponding subtree of ${\cal T}$ in ${\cal C'}$. By the way the new $z$-variables are introduced,
the monomial generated by ${\cal T'}$ is $\alpha\pi$ with $\alpha$ as the product of all the $z$-variables added to ${\cal T}$
to yield ${\cal T'}$. Since $\pi$ has degree $k$, ${\cal T}$ has $k-1$ many $\times$ gates. So,
${\cal T'}$ has $k-1$ new $\times$ gates along with $k-1$ many new $z$-variables that are added with respect to those $\times$ gates in
${\cal T}$. In addition, ${\cal T'}$ has $k$ terminal nodes representing $k$ individual
copies of $x$-variables in $\pi$. When such a terminal node is connected to a $+$ gate, then a new $\times$ gate is added along with a new
$z$-variable. Thus, the terminal nodes in ${\cal T}'$ can contribute at most $k$ additional $z$-variables.
Therefore, the degree of $\alpha$ is at most $2k-1$. Since all those $z$-variables are distinct, $\alpha$ is multilinear.

If $F'$ has a monomial $\alpha\pi$ such that $\alpha$ is a product of $z$-variables and $\pi$ is a product of $x$-variables,
then let ${\cal M}'$ be the subtree of ${\cal C'}$ that generates $\alpha\pi$. According to the construction of ${\cal C^*}$ and ${\cal C'}$,
removing all the $z$-variables
along with the newly added $\times$ gates from ${\cal M'}$ will result in a subtree ${\cal M}$ of ${\cal C}^*$
that generates $\pi$. Thereby, $\pi$ is a monomial in $F$.

 Assume that $F'$ has $\alpha_1 \pi$ and $\alpha_2\pi$ in its sum-product expansion,
where $\alpha_1$ and $\alpha_2$ are products of $z$-variables. Let ${\cal T}'_1$ and ${\cal T}'_2$
be the two subtrees in ${\cal C'}$ that generate $\alpha_1 \pi$ and $\alpha_2\pi$, respectively.
Since each of such subtrees in ${\cal C'}$  can be used once to generate one product in the sum-product expansion
of $F'$, we have ${\cal T}'_1 \not= {\cal T}'_2.$ Let ${\cal T}_1$ and ${\cal T}_2$ be the two respective subtrees of
${\cal T}'_1$ and ${\cal T}'_2$ in ${\cal C^*}$. By the ways of circuit reconstruction and
introduction of new $z$-variables,  ${\cal T}'_1 \not= {\cal T}'_2$ implies ${\cal T}_1 \not= {\cal T}_2$.
Note that ${\cal T}_1$ and ${\cal T}_2$ generates the same $\pi$. There are two cases for ${\cal T}_1$ and ${\cal T}_2$ to differ:
either ${\cal T}_1$ and ${\cal T}_2$ differ at a $\times$ gate $g$,
or they have the same $\times$ gates but differ at a terminal node $u$.
In the former case, the $z$-variables added with respect to  $g$ will make $\alpha_1$ and $\alpha_2$ different.
In the latter case, we assume without loss of generality that ${\cal T}_1$ has a terminal node $u$ but ${\cal T}_2$ does not.
In this case, the parent node $u'$ of $u$ has to be a $+$ gate. Hence, a new $z$-variable is added for the new $\times$ gate
between $u'$ and $u$. Therefore, this new $z$-variable makes $\alpha_1$ and $\alpha_2$ different.

Now, consider that $F'$ has two monomials $\alpha\pi$ and $\beta\phi$ such that, $\pi$ and $\phi$ are products of
$x$-variables and $\alpha$ and $\beta$ are products of $z$-variables. Let ${\cal H}'_1$ and ${\cal H'}_2$
be the subtrees in ${\cal C'}$  that generate $\alpha\pi$ and $\beta\phi$, respectively.
Again, according to the construction of ${\cal C^*}$
and ${\cal C'}$, removing all the $z$-variables
along with the newly added $\times$ gates from ${\cal H'}_1$ and ${\cal H'}_2$
will result in two subtrees ${\cal H}_1$ and ${\cal H}_2$ of ${\cal C}^*$
that generate $\pi$ and $\phi$, respectively.
When $\pi \not= \phi$, ${\cal H}_1$ and ${\cal H}_2$ are different subtrees.
Following a similar analysis in the above paragraph for ${\cal T}_1$ and ${\cal T}_2$ to be different,
we have  $\alpha \not= \beta.$
Also, since the $z$-variables in $\alpha$ corresponds to $\times$ gates in ${\cal H'}_1$
that do not repeat themselves because ${\cal H'}_1$ is a tree,
$\alpha$ is multilinear. Similarly, $\beta$ is also multilinear.

Combining the above analysis completes the proof for the lemma.
\end{proof}

\subsection{A Transformation}

In order to present the technique to transform the testing of $q$-monomials  to the testing of
multilinear monomials, we introduce one more definition
related to variable replacements.

\begin{definition}\label{def-2}
Let $q\ge 2$ be a fixed integer. Let $\pi = x^s$ for $1\le s\le q-1$. Consider
\begin{eqnarray}\label{eq-0}
r(\pi) & = & \prod^{s}_{i=1}(c_{i1}y_1+c_{i2}y_2+\cdots + c_{i(q-1)}y_{q-1}), \nonumber
\end{eqnarray}
where $c_{ij}$ are constants and $y_j$ are new variables, $1\le i\le s$ and $1\le j\le q-1$.
For $1\le s \le q-1$, let $\pi' = y_1y_2\cdots y_s$.
Define the coefficient matrix of $\pi'$ with respect to $r(\pi)$ as
$$
{\cal C}[\pi', r(\pi)]=\left(\begin{array}{cccc}
c_{11} & c_{12} & \cdots & c_{1s}\\
c_{21} & c_{22} & \cdots & c_{2s}\\
        &   & \cdots &   \\
c_{s1} & c_{s2} & \cdots & c_{ss}
\end{array}
\right).
$$
\end{definition}

{\bf Transformation:}\label{Trans}
For any given $n$-variate polynomial $F(x_1,x_2,\ldots,x_n)$ represented by a circuit ${\cal C}$,
we first carry out the circuit reconstruction as addressed in Subsection \ref{CR} to obtain a new circuit ${\cal C'}$
and let $F'(z_1, z_2, \ldots, z_h, x_1, x_2,\ldots, x_n)$ be the new polynomial represented by ${\cal C'}$.
The transformation through replacing $x$-variables works as follows:
For each variable $x_i$ and for each terminal node $u_j$
representing $x_i$ in circuit ${\cal C'}$,
select uniform random values $c_{ij\ell}$ from $Z_2$ and replace $x_i$ at the node
$u_j$ with
\begin{eqnarray}\label{trans-exp1}
r(x_i) &=& (c_{ij1}y_{i1} + c_{ij2}y_{i2} +\cdots+c_{ij(q-1)}y_{i(q-1)}).
\end{eqnarray}
Let
$$
G(z_1,\ldots,z_h,y_{11},\ldots,y_{1(q-1)},\ldots,y_{n1},\ldots,y_{n(q-1)})
$$
be the polynomial resulted from the above replacements for circuit ${\cal C'}$.

We need Lemmas \ref{lem3} and \ref{lem-0} in the following to help estimate the success probability of the transformation.

Consider the vector space $Z_2^n$.
For any vector $\vec{v}_1,\vec{v}_2,\ldots,\vec{v}_{k}\in Z_2^n$, $1\le k\le n$, let
$\mbox{span}(\vec{v}_1,\vec{v}_2,\ldots,\vec{v}_{k})$ denote the linear space generated by those $k$ vectors.
The following lemma follows directly from Lemma 6.3.1 of Blum and Kannan in \cite{blum95}.

\begin{lemma}\label{lem3}
\cite{blum95} Assume that $\vec{v}_1,\vec{v}_2,\ldots,\vec{v}_{k}$ are random vectors
uniformly chosen from $Z_2^n$, $1\le k\le n$ and $n\ge 1$. Let
$\mbox{Pr}[\vec{v}_1,\vec{v}_2,\ldots,\vec{v}_{k}]$ denote the probability that
$\vec{v}_1,\vec{v}_2,\ldots,\vec{v}_{k}$ are linearly independent. We have
$$
\mbox{Pr}[\vec{v}_1,\vec{v}_2,\ldots,\vec{v}_{k}] >0.28.
$$
\end{lemma}

Koutis  had a proof for $\mbox{Pr}[\vec{v}_1,\vec{v}_2,\ldots,\vec{v}_{k}] > \frac{1}{4}$,
which is contained in the proof for his Theorem 2.4 \cite{koutis08}.
But some careful examination will show that
there is a flaw in the analysis for $k=3$. Nevertheless, we present a proof in the following.

\begin{proof}
From the basis of linear algebra,
we know that $\mbox{span}(\vec{v}_1,\vec{v}_2,\ldots,\vec{v}_{k})$ has $2^k$ vectors and
any vector in $Z^n_2 - \mbox{span}(\vec{v}_1,\vec{v}_2,\ldots,\vec{v}_{k})$ is linearly independent of
$\vec{v}_1,\vec{v}_2,\ldots,\vec{v}_{k}.$ Note that $|Z^n_2| = 2^n$.
Therefore,
\begin{eqnarray}\label{eq-vec}
\mbox{Pr}[\vec{v}_1,\vec{v}_2,\ldots,\vec{v}_{k}]
&=&  \mbox{Pr}[\vec{v}_k \not\in \mbox{span}(\vec{v}_1,\vec{v}_2,\ldots,\vec{v}_{k-1})] \cdot
\mbox{Pr}[\vec{v}_1,\vec{v}_2,\ldots,\vec{v}_{k-1}] \nonumber \\
&=&(1-\frac{1}{2^{n-k +1}})\cdot \mbox{Pr}[\vec{v}_1,\vec{v}_2,\ldots,\vec{v}_{k-1}] \nonumber \\
&=& \prod^k_{i=1}(1-\frac{1}{2^{n-i+1}}) \nonumber \\
&\ge & \prod^{k}_{i=1}(1-\frac{1}{2^i}).
\end{eqnarray}
The last inequality holds because of $1\le k \le n$. For any $1\le k \le 40$, by simply carrying out
the computation for the right product of expression (\ref{eq-vec}), we obtain
\begin{eqnarray} \label{eq-vec-2}
\prod^{k}_{i=1}(1-\frac{1}{2^i}) &\ge & 0.288788 > 0.28,~~ 1\le k \le 40, \mbox{and} \\
\left(~\prod^{40}_{i=1}(1-\frac{1}{2^i})~\right) \cdot \frac{40}{41} &\ge & 0.281444 > 0.28.
\end{eqnarray}
It is obvious that $2^i > i^2$ for $i\ge 41$. Combining this with expressions (\ref{eq-vec}) and (5)
yields, for any $k>40$,
\begin{eqnarray}\label{eq-vec-3}
\prod^{k}_{i=1}(1-\frac{1}{2^i}) &=&
\left(~\prod^{40}_{i=1}(1-\frac{1}{2^i})~\right) \cdot
\left(~\prod^{k}_{i=41}(1-\frac{1}{2^i})~\right) \nonumber \\
&\ge & \left(~\prod^{40}_{i=1}(1-\frac{1}{2^i})~\right) \cdot
\left(~\prod^{k}_{i=41}(1-\frac{1}{i^2})~\right) \nonumber \\
&=& \left(~\prod^{40}_{i=1}(1-\frac{1}{2^i})~\right) \cdot
\left(~\prod^{k}_{i=41}(\frac{(i-1)(i+1)}{i^2})~\right) \nonumber \\
&=& \left(~\prod^{40}_{i=1}(1-\frac{1}{2^i})~\right) \cdot \frac{40}{41}
\cdot \frac{k+1}{k}\nonumber \\
&\ge & 0.28 \cdot \frac{k+1}{k}\nonumber \\
&\ge & 0.28
\end{eqnarray}
The complete proof is then derived from expressions (\ref{eq-vec-2}) and (\ref{eq-vec-3}).
\end{proof}

\begin{lemma}\label{lem-0}
For any integer matrix  ${\cal A} = (a_{ij})_{n\times n}$, we have
\begin{eqnarray}
\mbox{perm}({\cal A}) \imod 2 & = & \mbox{det}({\cal A}) \imod 2.
\end{eqnarray}
\end{lemma}

\begin{proof}
Let $\lambda$ be any permutation of $\{1,2,\ldots,n\}$, and $\mbox{sign}(\lambda)$ be the sign of
the permutation $\lambda$. Since for any integer $b$, $b \equiv -b \imod 2$, we have
\begin{eqnarray}
\mbox{det}({\cal A}) \imod 2 & = &
\left(~\sum_{\lambda} (-1)^{\mbox{sign}(\lambda)}a_{1\lambda(1)}a_{2\lambda(2)}\cdots a_{n\lambda(n)}~\right) \imod 2
\nonumber \\
& = & \left(~\sum_{\lambda} a_{1\lambda(1)}a_{2\lambda(2)}\cdots a_{n\lambda(n)}~\right) \imod 2
\nonumber \\
& = & \mbox{perm}({\cal A}) \imod 2. \nonumber
\end{eqnarray}
\end{proof}

It is obvious that the above lemma can be easily extended to any field of characteristic 2.
We are now ready to estimate the success probability of the transformation.

\begin{lemma}\label{rtm-lem2}
Assume that the variable replacements are carried out over a field ${\cal F}$ of characteristic $2$ (e.g., $Z_2$).
If a given $n$-variate polynomial $F(x_1,x_2,\ldots,x_n)$ that is represented by
a tree-like circuit ${\cal C}$ has a $q$-monomial of $x$-variables with degree $k$, then, with a probability at least $0.28^k$,
$G$ has a unique multilinear monomial $\alpha \pi$ such that $\pi$ is a degree $k$ multilinear monomial of $y$-variables
and $\alpha$ is a multilinear monomial of $z$-variables with degree $\le 2k-1$. If $F$ has no $q$-monomials, then
$G$ has no multilinear monomials of $y$-variables, i.e., $G$ has no monomials of the format $\beta \phi$ such that
$\beta$ is a multilinear monomial of $z$-variables and $\phi$ is a multilinear monomial of $y$-variables.
\end{lemma}

\begin{proof}
We first show the second part of the lemma, i.e.,
if $F$ has no $q$-monomials, then $G$ has no multilinear monomials of $y$-variables.
Suppose otherwise that $G$ has a multilinear monomial $\beta \phi$.
Let $\phi = \phi_1\phi_2\cdots \phi_s$ such that $\phi_j$ is the product of all the $y$-variables
in $\phi$ that are used to replace the variable $x_{i_j}$, and let $\mbox{deg}(\phi_j) = d_j$, $1\le j\le  s$.
Consider the subtree ${\cal T}'$ of ${\cal C}'$ that generates $\beta\phi$ when the $x$-variables are replaced by
a linear sum of $y$-variables according to expression (\ref{trans-exp1}). Then, the subtree ${\cal T}$ in ${\cal C}^*$ that corresponds to
${\cal T}'$ in ${\cal C}'$ computes a monomial $\pi = x_{i_1}^{d_i}x_{i_2}^{d_2}\cdots x^{d_s}_{i_s}$ and
$\phi$ is a multilinear monomial in the expansion of the replacement $r(\pi)$, which
is obtained by replacing each occurrence of $x$-variable with a linear sum of $(q-1)$ many $y$-variables
by expression $(\ref{trans-exp1})$. If there is one $d_j$ such that
$d_j\ge q$, then let us look at the replacements for $x_{i_j}^{d_j}$, denoted as
$$
r(x_{i_j}^{d_j}) = \prod^{d_j}_{t = 1} (c_{t1}y_{1}+c_{t2}y_{2}+\cdots + c_{t(q-1)}y_{(q-1)}).
$$
Since $d_j\ge q$, by the pigeon hole principle, the expansion of the above $r(x_{i_j}^{d_j})$ has no multilinear monomials.
Thereby, we must have $1\le d_j\le q-1$, $1\le j\le s$. Hence, $\pi$ is a $q$-monomial in $F$, a contradiction to our assumption at
the beginning. Therefore, when $F$ has no $q$-monomials, then $G$ must not have any multilinear monomials of $y$-variables.

We now prove the first part of the lemma.
Suppose $F$ has a $q$-monomial $\pi = x_{i_1}^{s_1}x_{i_2}^{s_2}\cdots x_{i_t}^{s_t}$ with $1\le s_j \le q-1$, $1\le j\le t$, and
$k = \mbox{deg}(\pi)$.
By Lemma \ref{rtm-lem1}, $F'$ has at least one monomial corresponding to $\pi$. Moreover, each of such monomials
has a format $\alpha \pi$ such that $\alpha$ is a unique multilinear monomials of $z$-variables with
$\mbox{deg}(\alpha) \le 2k-1$. Let $\beta = \alpha \pi$ be one of such monomials.
Consider the subtree ${\cal T'}$ of ${\cal C}'$ that generates $\beta$. Based on the construction of
${\cal C}'$, ${\cal T'}$ has $s_j$ terminal nodes representing $s_j$ occurrences of $x_{i_j}$ in $\pi$, $1\le j\le t$.
By variable replacements in expression (\ref{trans-exp1}), $\beta$ becomes $r(\beta)$ as follows:
\begin{eqnarray}\label{rtm-exp3}
 r(\beta) &=& \alpha r(\pi) \nonumber \\
 &=& \alpha~\prod^t_{\ell=1} \left[~ \prod^{s_\ell}_{j=1}(c_{\ell j1}y_{\ell1}+c_{\ell j2}y_{\ell 2}+\cdots + c_{\ell j(q-1)}y_{\ell (q-1)})~\right],
\end{eqnarray}
where each occurrence $j$ of $x_{i_\ell}$ is replaced by $(c_{\ell j1}y_{\ell1}+c_{\ell j2}y_{\ell 2}+\cdots + c_{\ell j(q-1)}y_{\ell (q-1)}).$
For $1\le \ell \le t$, let $\pi_\ell = x^{s_\ell}_{i_\ell}$,  and
\begin{eqnarray}\label{rtm-exp4}
 r(\pi_\ell) &=&
\prod^{s_\ell}_{j=1}(c_{\ell j1}y_{\ell1}+c_{\ell j2}y_{\ell 2}+\cdots + c_{\ell j(q-1)}y_{\ell (q-1)}).
\end{eqnarray}
Since $1\le s_{\ell} \le q-1$, by expression (\ref{rtm-exp4}), $r(\pi_\ell)$ has a multilinear monomial $\pi'_\ell$ with coefficient
$c_\ell$ such that
\begin{eqnarray} \label{rtm-exp5}
\pi'_\ell & = & y_{\ell1} y_{\ell 2} \cdots y_{\ell s_{\ell}}, \mbox{ and }
\end{eqnarray}
\begin{eqnarray} \label{rtm-exp6}
c_\ell &=& \mbox{perm}(C[\pi'_\ell, r(\pi_\ell)]),
\end{eqnarray}
where the coefficient matrix, as defined in Definition \ref{def-2}, is
$$
{\cal C}[\pi'_\ell, r(\pi_\ell)]=\left(\begin{array}{cccc}
c_{\ell 11} & c_{\ell 12} & \cdots & c_{\ell 1s_\ell}\\
c_{\ell 21} & c_{\ell 22} & \cdots & c_{\ell 2s_\ell}\\
        &   & \cdots &   \\
c_{\ell s_\ell 1} & c_{\ell s_\ell 2} & \cdots & c_{\ell s_\ell s_\ell}
\end{array}
\right).
$$
Since the field ${\cal F}$ has characteristic $2$ and all the entries in the coefficient are $0/1$ values,
we have by Lemma \ref{lem-0}
$$
\mbox{perm}(C[\pi'_\ell, r(\pi_\ell)]) = \mbox{det}(C[\pi'_\ell, r(\pi_\ell)]).
$$
Because each row of $C[\pi'_\ell, r(\pi_\ell)]$ is a uniform random vector in $Z^{s_\ell}_2$,
by Lemma \ref{lem3}, with a probability of at least $0.28$, those row vectors are linearly independent, implying
$\mbox{det}(C[\pi'_\ell, r(\pi_\ell)]) =1$. Hence, by expressions (\ref{rtm-exp4}),  (\ref{rtm-exp5}) and (\ref{rtm-exp6}),
with a probability at least $0.28$, $r(\pi_\ell)$ has a multilinear monomial $\pi'_\ell$.
By expression (\ref{rtm-exp3}), with a probability at least $0.28^t\ge 0.28^k$,
$\alpha r(\pi)$ has a desired multilinear monomial $\alpha \pi'_1\pi'_2\cdots \pi'_t$.
\end{proof}

\section{Randomized Testing of $q$-monomials}

Let $d = \log_2 (2k-1) + 1$ and ${\cal F} = \mbox{GF}(2^d)$ be a finite field of $2^d$ many elements.
We consider the group algebra
${\cal F}[Z^k_2]$. Please note that the field ${\cal F} = \mbox{GF}(2^d)$ has characteristic $2$.
This implies that, for any given element $w \in {\cal F}$, adding $w$ for any even number of times yields $0$.
For example, $w + w = 2w = w+w+w+w = 4w = 0.$

The algorithm RandQMT for testing whether any given $n$-variate polynomial $F(x_1,x_2,\ldots,x_n)$ that is presented by
a tree-like circuit ${\cal C}$ has a $q$-monomial of degree $k$ is given in the following.

\begin{quote}
\noindent{\bf Algorithm RandQMT} (\underline{Rand}omized $\underline{q}$-\underline{M}onomials \underline{T}esting):
\begin{description}
\item[1.] As described in Subsection \ref{CR},
reconstruct the circuit ${\cal C}$ to obtain ${\cal C}^*$ that computes the same polynomial
$F(x_1,x_2,\ldots,x_n)$ and then introduce new $z$-variables to ${\cal C}^*$
to obtain the new circuit ${\cal C'}$ that computes $F'(z_1,z_2,\ldots, z_h,x_1,x_2,\ldots, x_n)$.

\item[2.] Repeat the following loop for at most $(\frac{1}{0.28})^k$ times.
\begin{description}
\item[2.1.] For each variable $x_i$ and for each terminal node $u_j$
representing $x_i$ in circuit ${\cal C'}$,
select uniform random values $c_{ij\ell}$ from $Z_2$ and replace $x_i$ at the node
$u_j$ with
\begin{eqnarray}\label{rtm-exp1}
\hspace{-0.3in} r(x_i) &=& (c_{ij1}y_{i1} + c_{ij2}y_{i2} +\cdots+c_{ij(q-1)}y_{i(q-1)}).
\end{eqnarray}
Let
$$G(z_1,\ldots,z_h,y_{11},\ldots,y_{1(q-1)},\ldots,y_{n1},\ldots,y_{n(q-1)})
$$
be the polynomial resulted from the above replacements for circuit ${\cal C'}$.

\item[2.2.] Select uniform random vectors $\vec{v}_{ij}
\in Z^k_2-\{\vec{v}_0\}$, and replace the variable $y_{ij}$ with
$(\vec{v}_{ij} + \vec{v}_0)$, $1\le i \le n$ and $1\le j\le q-1$.

\item[2.3.] Use ${\cal C'}$ to calculate
\begin{eqnarray}\label{rtm-exp2}
G' &=& G(z_1,\ldots,z_h,(\vec{v}_{11}+\vec{v}_0),\ldots,(\vec{v}_{1(q-1)}+\vec{v}_0),
\ldots,\nonumber \\
& & ~~~~(\vec{v}_{n1}+\vec{v}_0),\ldots,(\vec{v}_{n(q-1)}+\vec{v}_0)) \nonumber \\
& = & \sum_{j=1}^{2^k} f_j(z_1,\ldots,z_h) \cdot \vec{v}_j,
\end{eqnarray}
where each $f_j$ is a polynomial of degree $\le 2k-1$ over the finite
field ${\cal F}=\mbox{GF}(2^d)$, and $\vec{v}_j$ with $1\le j\le 2^k$ are the $2^k$ distinct vectors in
$Z^k_2$.

\item[2.4.] Perform polynomial identity testing with the
Schwartz-Zippel algorithm \cite{motwani95} for every
$f_{j}$ over ${\cal F}$.
Return {\em "yes"} if one of those polynomials is not
identical to zero.
\end{description}
\item[3.] Return {\em "no"} if no {\em "yes"} has been returned in the loop.
\end{description}
\end{quote}

It should be pointed out that the actual implementation of Step 2.3 would be running the
Schwartz-Zippel algorithm concurrently for all $f_j$, $1\le j\le 2^k$,
utilizing the circuit ${\cal C'}$. If one of those polynomials is not
identical to zero, then the output of $G'$ as computed by circuit ${\cal C'}$ is not zero.

The group algebra technique established by Koutis \cite{koutis08} assures
the following two properties:

\begin{lemma}\label{rtm-lem3}
(\cite{koutis08})~  Replacing all the variables $y_{ij}$
in $G$ with group algebraic elements $\vec{v}_{ij}+\vec{v}_0$ will make all
monomials $\alpha \pi$ in $G$ become zero, if $\pi$ is non-multilinear with respect to $y$-variables. Here,
$\alpha$ is a product of $z$-variables.
\end{lemma}

\begin{proof}
Recall that ${\cal F}$ has characteristic $2$.
For any $\vec{v} \in Z^k_2$, in the group algebra ${\cal F}[Z^k_2]$,
\begin{eqnarray}\label{rt-1}
(\vec{v}+\vec{v}_0)^2 &=& \vec{v}\cdot\vec{v} + 2\cdot\vec{v}\cdot\vec{v}_0 + \vec{v}_0\cdot\vec{v}_0 \nonumber \\
 &=&\vec{v}_0+ 2\cdot\vec{v} + \vec{v}_0 \nonumber \\
 &=& 2\cdot \vec{v}_0 + 2\cdot\vec{v} = {\bf 0}.
\end{eqnarray}
Thus, the lemma follows directly from expression (\ref{rt-1}).
\end{proof}

\begin{lemma}\label{rtm-lem4}
(\cite{koutis08})~ Replacing all the variables $y_{ij}$ in $G$
with group algebraic elements $\vec{v}_{ij}+\vec{v}_0$ will make any monomial $\alpha \pi$
to become zero,  if and only if  the vectors $\vec{v}_{ij}$  are linearly dependent in the vector space $Z^k_2$.
Here, $\pi$ is a multilinear monomial of $y$-variables and $\alpha$ is a product of $z$-variables.
Moreover, when $\pi$ becomes non-zero after the replacements,
it will become the sum of all the vectors in the linear space spanned by those vectors.
\end{lemma}

\begin{proof}
The analysis below gives a proof for this lemma.
Suppose $V$  is a set of  linearly dependent vectors
in $Z^k_2$. Then, there exists a nonempty subset $T \subseteq V$ such that $\prod_{\vec{v}\in T} \vec{v}= \vec{v}_0$.
For any $S\subseteq T$, since
$\prod_{\vec{v}\in T} \vec{v} =
(\prod_{\vec{v}\in S} \vec{v}~) \cdot (\prod_{\vec{v}\in T-S} \vec{v}~)$, we have
$\prod_{\vec{v}\in S} \vec{v} = \prod_{\vec{v}\in T-S} \vec{v}$.
Thereby, we have
\begin{eqnarray}
\prod_{\vec{v}\in T }(\vec{v} + \vec{v}_0)
&=& \sum_{S\subseteq T} \prod_{\vec{v}\in S}\vec{v} = {\bf 0}, \nonumber
\end{eqnarray}
since every $\prod_{\vec{v}\in S}\vec{v}$ is paired by the same
$\prod_{\vec{v}\in T-S}\vec{v}$ in the sum above and the addition of the pair
is annihilated because ${\cal F} $ has characteristic $2$. Therefore,
\begin{eqnarray}
\prod_{\vec{v}\in V}(\vec{v} + \vec{v}_0)
&=& \left(~\prod_{\vec{v}\in T} (\vec{v}+\vec{v}_0)\right) \cdot \left(~\prod_{\vec{v}\in V-T}(\vec{v}+\vec{v}_0)\right)  \nonumber \\
& =&  0 \cdot \left(~\prod_{\vec{v}\in V-T}(\vec{v}+\vec{v}_0)\right) = {\bf 0}. \nonumber
\end{eqnarray}

Now consider that vectors in $V$ are linearly independent.
For any two distinct subsets $S, T\subseteq V$, we must have
$\prod_{\vec{v}\in T} \vec{v} \not=\prod_{\vec{v}\in S} \vec{v}$, because otherwise
vectors in $S \cup T - (S \cap T)$ are linearly dependent, implying that vectors in $V$ are linearly dependent.
Therefore,
\begin{eqnarray}
\prod_{\vec{v}\in V}(\vec{v} + \vec{v}_0)
&=& \sum_{T\subseteq V} \prod_{\vec{v}\in T} \vec{v} \nonumber
\end{eqnarray}
is the sum of all the $2^{|V|}$ distinct vectors spanned by $V$.
\end{proof}

\begin{theorem}\label{thm-rtm}
Let $q>2$ be any fixed integer and $F(x_1,x_2,\ldots,x_n)$ be an $n$-variate polynomial
represented by a tree-like circuit ${\cal C}$ of
size $s(n)$. Then the randomized algorithm RandQMT can decide whether $F$ has a $q$-monomial
of degree $k$ in its sum-product expansion in time $O^*(7.15^k s^2(n))$.
\end{theorem}

For applications, we often require that the size of a given circuit is a polynomial in $n$.
in such cases, the upper bound in the theorem becomes $O^*(7.15^k)$.

\begin{proof}
From the introduction of the new $z$-variables to the circuit ${\cal C'}$,
it is easy to see that every monomial in $F'$ has the format $\alpha \pi$, where $\pi$ is a product of
$x$-variables and $\alpha$ is a product of $z$-variables. Since only $x$-variables are replaced by respective linear sums of
new $y$-variables as specified in expression (\ref{rtm-exp1}) (or expression (\ref{trans-exp1})), monomials in $G$ have the format
$\beta \phi$, where $\phi$ is a product of
$y$-variables and $\beta$ is a product of $z$-variables.

Suppose that $F$ has no $q$-monomials. By Lemma \ref{rtm-lem2}, $G$ has no monomials $\beta\phi$ such
that $\phi$ is a multilinear monomial of  $y$-variables and $\beta$ is a product of $z$-variables. In other words,
for every monomial $\beta\phi$ in $G$, the $y$-variable product $\phi$ must not be multilinear.
Moreover, by Lemma \ref{rtm-lem3}, replacing $y$-variables will make $\phi$ in every monomial $\beta\phi$ in $G$ to become zero. Hence,
the replacements will make  $G$ to become zero and so the algorithm RandQMT will return {\em "no"}.

Assume that $F$ has a $q$-monomial of degree $k$. By Lemma \ref{rtm-lem2},
with a probability at least $0.28^k$, $G$ has a monomial $\beta \phi$ such that
$\phi$ is a $y$-variable multilinear monomial of degree $k$ and
$\beta$ is a $z$-variable multilinear monomial of degree $\le 2k -1$.
It follows from Lemma \ref{lem3}, a list of uniform vectors from $Z^{k}_2$ will
be linearly independent with a probability at least $0.28$.
By Lemma \ref{rtm-lem4}, with a probability at least $0.28$, the multilinear monomial
$\phi$ will not be annihilated by the group algebra replacements at Steps 2.2 and 2.3.
Precisely, with a probability at least $0.28$, $\beta\phi$ will become
\begin{eqnarray}\label{rtm-exp7}
\lambda(\beta\phi) & = & \sum^{2^k}_{i=1} \beta \vec{v}_i,
\end{eqnarray}
where $\vec{v}_i$ are distinct vectors in $Z^k_2$.

Let ${\cal S}$ be the set of all those multilinear monomials $\beta\phi$ that survive
the group algebra replacements for $y$-variables in $G$. Then,
\begin{eqnarray}\label{rtm-exp8}
G' &= &G(z_1,\ldots,z_h,(\vec{v}_{11}+\vec{v}_0),\ldots,(\vec{v}_{1(q-1)}+\vec{v}_0),\ldots,\nonumber \\
&&~~~(\vec{v}_{n1}+\vec{v}_0),\ldots,(\vec{v}_{n(q-1)}+\vec{v}_0)) \nonumber \\
&=& \sum_{\beta\phi \in {\cal S}} \lambda(\beta\phi)  \nonumber \\
&=&  \sum_{\beta\phi\in {\cal S}} \left(\sum^{2^k}_{i=1}\beta \vec{v}_i\right) \nonumber \\
&= & \sum_{j=1}^{2^k} \left(\sum_{\beta\phi\in {\cal S}} \beta \right) \vec{v}_j
\end{eqnarray}
Let
\begin{eqnarray}
f_j(z_1,\ldots,z_h) & = & \sum_{\beta\phi\in {\cal S}} \beta. \nonumber
\end{eqnarray}
By Lemmas \ref{rtm-lem2} and \ref{rtm-lem3}, the degree of $\beta$ is at most $2k-1$.
Hence, the coefficient polynomial $f_j$ with respect to $\vec{v}_j$ in $G'$ after the algebra replacements
has degree $\le 2k-1$. Also, by Lemma \ref{rtm-lem2},
$\beta$ is unique with respect to every $\phi$ for each monomial $\beta\phi$ in $G$.
Thus, the possibility of a {\rm "zero-sum"} of coefficients from different
surviving monomials is completely avoided during the construction
of $f_j$. Therefore, conditioned on that ${\cal S}$ is not empty,
$F'$ must not be identical to zero, i.e., there exists at least one
$f_j$ that is not identical to zero. At Step 2.4, we use the
randomized algorithm by Schwartz-Zippel \cite{motwani95}
to test whether $f_j$ is identical to zero. It is known that
this testing can be done with a probability at least
$\frac{2k-1}{|{\cal F}|} = \frac{1}{2}$
in time polynomially in $s(n)$ and $\log_2 |{\cal F}| = 1+ \log_2 (2k-1)$. Since ${\cal S}$ is
not empty with a probability at least $0.28$, the success probability
of testing whether $G$ has a degree $k$ multilinear monomial is at
least $0.28 \times \frac{1}{2} > \frac{1}{8}$, under the condition that $G$ has at least one degree $k$ multilinear monomial.

Summarizing the above analysis,
when $F$ has a $q$-monomial of degree $k$
with a probability at least $0.28^k$, $G$ has a degree $k$ multilinear monomial $\phi$ of $y$-variables
in the format $\beta\phi$ with coefficient
$\beta$ that is a multilinear monomial of $z$-variables with degree $\le 2k-1$.
Thus, the probability that $G$ does not have any degree $k$  multilinear monomials of $y$-variables in the aforementioned format
$\beta\phi$ in its sum-product
expansion during any of the $\left(\frac{1}{0.28}\right)^k$ loop iterations is at most
$$
\left(1- (0.28)^k\right)^{(\frac{1}{0.28})^k} \le \frac{1}{e}.
$$
This implies that the probability that $G$ has at least one degree $k$ multilinear monomial during at least one of the
$\left(\frac{1}{0.28}\right)^k$ loop iterations is at least
$$
1- \frac{1}{e}.
$$
When $G$ has at least one degree $k$
multilinear monomial $\phi$ of $y$-variables in the format $\beta\phi$ as described above,
the group algebra replacement technique and the Schwartz-Zippel polynomial identity testing algorithm as analyzed above
will detect this with a probability at least $\frac{1}{8}$.
Therefore, when $F$ has one $q$-monomial in its sum-product expansion, with a probability at least
$$
\frac{1}{8} \times \left(1-\frac{1}{e}\right),
$$
algorithm RandQMT will detect this.

Finally, we address the issues about how to calculate $G'$ and the
time needed to do so. Naturally, every element in the group
algebra ${\cal F}[Z^k_2]$ can be represented by a vector in
$Z^{2^k}_2$. Adding two elements in ${\cal F}[Z^k_2]$ is equivalent to
adding the two corresponding vectors in $Z_2^{2^k}$, and the
latter can be done in $O(2^k)$ time via component-wise sum.
In addition, multiplying two elements in ${\cal F}[Z^k_2]$ is
equivalent to multiplying the two corresponding vectors in
$Z_2^{2^k}$, and the latter can be done in $O(k2^k\log_2|{\cal F}|)=O(k^22^k)$ with
the help of a similar Fast Fourier Transform style algorithm as in
Williams \cite{williams09}. Calculating $G'$ consists of $n * s^2(n)$
arithmetic operations of either adding or multiplying two elements
in ${\cal F}[Z^k_2]$ based on the circuit $C'$. Hence, the total
time needed is $O(n*s^2(n) k^2 2^k)=O^*(2^k s^2(n))$. At Step 2.4, we run the
Schwartz-Zippel algorithm on $G'$ to
simultaneously test whether there is one $f_j$ such that $f_j$ is
not identical to zero.
 The total time for the entire
algorithm is $O^*(2^k s^2(n) \cdot (\frac{1}{0.28})^k)$.
Since
$$
2\times \frac{1}{0.28} =2\times \frac{100}{28} <7.15,
$$
the time complexity of algorithm RandQMT is bounded by
$O^*(7.15^k s^2(n)).$
\end{proof}

\section{Concluding Remarks}
The group algebra approaches to testing multilinear monomials \cite{koutis08,williams09}
and $q$-monomials for prime $q$ \cite{chen11b,chen12b} rely on
the property that $Z_2$  and $Z_q$ are fields for primes $q > 2$.
These approaches are not applicable to the general case of testing $q$-monomials,
since $Z_q$ is no longer a field when $q$ is not prime.
In this paper, we have developed a variable replacement technique and a new way to reconstruct a given circuit.
When the two are combined,  they help us transform the $q$-monomial testing problem to
the multilinear monomial testing problem
in a randomized setting. We have also proved that the transformation has the desired
success probability to warrant its application to the design of our new algorithm.

It should be pointed out that the time complexity of the randomized
$q$-monomial testing algorithm obtained in \cite{chen11} runs in time
$O^*(q^k)$ for prime $q\ge 2$, when the size of the circuit is a polynomial in $n$.
Algorithm  $\mbox{RandQMT}$ runs in time $O^*(7.15^k)$,
hence it significantly improves the time complexity of the algorithm in \cite{chen11} for prime $q > 7$.

\section*{Acknowledgments}

Shenshi is supported by Dr. Bin Fu's NSF CAREER Award, 2009 April 1 to 2014 March 31.
Yaqing is supported by a UTPA Graduate Assistantship. Part of Quanhai's work was
done while he was visiting the Department of Computer Science at the University of
Texas-Pan American.



\end{document}